\documentclass[conference]{IEEEtran}

\usepackage{xcolor}
\usepackage{cite}

\ifCLASSINFOpdf
  \usepackage[pdftex]{graphicx}

\else

\fi

\usepackage[cmex10]{amsmath}
\usepackage{algorithmic}
\usepackage{array}
\usepackage{url}
\usepackage{amsthm}
\newtheorem{proposition}{\bf Proposition}

\newtheorem{definition}{\bf Definition}
\newtheorem{remark}{Remark}

\newcommand{\ve}[1]{\boldsymbol{#1}}

\newcommand{\cm}[1]{\textcolor{black}{#1}}

\newcommand{\argmax}{\operatornamewithlimits{argmax}}
\newcommand{\argmin}{\operatornamewithlimits{argmin}}
\newcolumntype{I}{!{\vrule width 1.2pt}}
\makeatletter
\def\hlinewd#1{%
\noalign{\ifnum0=`}\fi\hrule \@height #1 %
\futurelet\reserved@a\@xhline}
\usepackage{tikz}

\IEEEoverridecommandlockouts
\usepackage{fancyhdr}
\makeatletter
\let\ps@plain\ps@fancy
\makeatother

\begin{document}
\title{Game-theoretic Demand-side Management Robust to Non-Ideal Consumer Behavior in Smart Grid}

\author{\IEEEauthorblockN{Chathurika P. Mediwaththe}
\IEEEauthorblockA{University of New South Wales and NICTA, Australia\\
Email: Chathurika.Mediwaththe@nicta.com.au}
\and
\IEEEauthorblockN{David B. Smith}
\IEEEauthorblockA{NICTA and Australian National University, Australia\\
Email: David.Smith@nicta.com.au}}

\maketitle
\thispagestyle{fancy} 

\begin{abstract}
This paper investigates effects of realistic, non-ideal, decisions of energy users as to whether to participate in an energy trading system proposed for demand-side management of a residential community. The energy trading system adopts a non-cooperative Stackelberg game between a community energy storage (CES) device and users with rooftop photovoltaic panels where the CES operator is the leader and the users are the followers. Participating users determine their optimal energy trading starting time to minimize their personal daily energy costs while subjectively viewing their opponents' actions. Following a non-cooperative game, we study the subjective behavior of users when they decide on energy trading starting time using prospect theory. We show that depending on the decisions of participating-time, the proposed energy trading system has a unique Stackelberg equilibrium at which the CES operator maximizes their revenue while users minimize their personal energy costs attaining a Nash equilibrium. Simulation results confirm that the benefits of the energy trading system are robust to decisions of participating-time that significantly deviate from complete rationality.
\end{abstract}

\IEEEpeerreviewmaketitle

\section{Introduction}
Demand-side management helps utilities to regulate increasing energy demand by utilizing existing power grid infrastructure. Recent efforts of demand-side management include load-shifting methods, load curtailing methods and energy conservation strategies \cite{MohensianGT}. Distributed energy resources such as energy storage devices and renewable energy resources provide vast opportunities for demand-side management by storing extra energy generated by renewable resources that can be dispatched to support peak energy demand.

In general, effectiveness of consumer-driven demand-side management methods depends on active participation of users. However, in the long run, users may change their participating behavior leading to unexpected outcomes such as lower peak energy reduction and economic benefits. Therefore, designing successful demand-side management approaches have often been challenging with volatile user behavior \cite{VolUserBehave}.

In this paper, we investigate impacts of realistic energy user behavior, which is not completely rational, on a decentralized energy trading system proposed to regulate electricity demand of a residential community. In the energy trading system, users with photovoltaic (PV) energy generation can decide to participate across time to trade energy with a community energy storage (CES) device. First, we elaborate a non-cooperative Stackelberg game to study the energy trading between the CES operator and participating users where the CES operator acts as the leader and the users are their followers. Then we develop another non-cooperative game between users to explore their behavior in determining optimal energy trading starting times that minimize personal daily energy costs under two different user-behavioral models: expected utility theory and prospect theory. The contributions of this work are:
\begin{itemize}
\item With time-varying subsets of active participating users that depend on their decisions of participating-time, the energy trading system attains a unique Stackelberg equilibrium across time where the CES operator maximizes revenue while users minimize energy costs.
\item Benefits of the energy trading system are robust to users' participating-time strategies that significantly deviate from complete rationality.
\end{itemize}

Game-theoretic demand-side management methods have been widely investigated in literature \cite{Quanyan,MohensianGT,Atzeni, Hung,Adika,Peng}. These studies assume that users act rationally and ideally obeying the strategies predicted by game-theoretic systems. However, social studies have proved that the rationality assumption of game theory can be violated in real world when users face uncertainty in decision making \cite{Kahneman}. Abundant research using prospect theory has shown how real life user behavior contravenes the conventional game theoretic rationality assumption \cite{Wang2,Wang3}. In \cite{Wang}, a prospect theoretic study for a load-shifting approach showed that deviations of users' decisions to participate from conventional game-theoretic decisions result in significantly different outcomes. In contrast to \cite{Wang}, we apply prospect theory to study users' behavior of choosing to participate across time in a Stackelberg game-theoretic energy trading system that does not intend to shift regular energy consumption of users. In this regard, we show that the outcomes of the energy trading system are indistinguishable under both prospect theory and expected utility theory, even though users' decisions to choose to participate differ between the two models. The Stackelberg game-theoretic energy trading system between a CES device and users in \cite{TCNS} assumes users participate from the beginning of day and hence the number of users remain consistent over time. Here, we extend the Stackelberg energy trading system to study users' decisions of selecting energy trading starting times incorporating prospect theory. The CES-user Stackelberg game in this paper differs from that in \cite{TCNS} because the number of active participating users is time-variant depending upon each user's decision of choosing an energy trading starting time.

\section{System configuration}\label{Sec1}
\subsection{Demand-side model}
The community consists of two types of energy users: participating users $\mathcal{P}~(|\mathcal{P}|=I)$ and non-participating users $\mathcal{N}~(|\mathcal{N}|=N)$. The users $\mathcal{P}$ have rooftop PV panels and they are the players in the energy trading optimization who trade energy with the grid and the CES device. The users $\mathcal{N}$ are conventional grid users without behind-the-meter energy generation and are not players in the energy trading optimization. Depending on net PV energy after consuming, the users $\mathcal{P}$ are classified into surplus users $\mathcal{S}_t$ and deficit users $\mathcal{D}_t$ those are time-dependent. For the energy trading optimization, the entire control time period $M$, usually a day, is partitioned into $K$ number of equal time slots with granularity of $\Delta$. We assume that PV power generation and demand forecasts of the following day are available to the users $\mathcal{P}$ to decide their day-ahead energy trading strategies. If $g_{n,t}$ and $e_{n,t}$ are the PV energy and the regular energy demand of user $n\in \mathcal{P}$ at time $t\in M$, respectively, then they sell/buy energy amount $x_{n,t}$ to/from the CES device at time $t$ such that,
\begin{equation}
x_{n,t}=l_{n,t}+(g_{n,t}-e_{n,t}),\label{eq id:2}
\end{equation}
where $l_{n,t}$ is the grid energy consumption of the user. Note that $l_{n,t}>0$ when the user buys energy from the grid and $l_{n,t}<0$ when the user sells energy to the grid. If the surplus energy of the user $n$ is $s_{n,t}=g_{n,t}-e_{n,t}$, each user $i\in \mathcal{S}_t$ sells energy to the CES device and user $j\in \mathcal{D}_t$ buys energy from the CES device such that,
\begin{equation}
\begin{split}
0\leq x_{i,t}\leq s_{i,t},\\
s_{j,t}\leq x_{j,t}\leq 0.\label{eq id:4}
\end{split}
\end{equation}
\subsection{Energy Storage Model}
The CES operator trades $l_{Q,t}$ energy with the grid at each time $t$ where $l_{Q,t}>0~(<0)$ if the CES device is charged (discharged). Here, we use the same CES model given in \cite{TCNS} that is similar to the energy storage model in \cite{Atzeni}. In this regard, per-slot energy trading amounts are given as $x_{n,t}=x_{n,t}^+-x_{n,t}^-$ and $l_{Q,t}=l_{Q,t}^+-l_{Q,t}^-$ where $x_{n,t}^+$ and $l_{Q,t}^+$ are the per-slot charging energy amounts, and $x_{n,t}^-$ and $l_{Q,t}^-$ are the per-slot discharging energy amounts. We define a charging efficiency $0<\beta^+\leq1 $, a discharging efficiency $\beta^-\geq 1$ and a leakage rate $0<\tau\leq 1$ for the energy storage. Denoting $q_0$ is the charge level at the beginning of day, the energy capacity limit of the CES device gives,
\begin{equation}
\ve{0}\prec q_0 \ve{\kappa}+\ve{\Gamma} \left[\ve{\mathcal{L}^+},-\ve{\mathcal{L}^-}\right]\ve{\beta}\preceq \ve{B}, \label{eq:id6}
\end{equation}
where $\ve{B} \in \Re^{K\times 1}$ with elements of maximum energy capacity of the CES device $B$. $\ve{\kappa}\in \Re^{K\times 1}$ with elements $[\ve{\kappa}]_l=\tau^l$ and the $(l,m)$ entry of the lower triangular matrix $\ve{\Gamma} \in\Re^{K\times K}$ is $[\ve{\Gamma}]_{l,m}=\tau^{l-m}$. $\ve{\beta}=[\beta^+,\beta^-]^T$,~$\ve{0}$ is the $K\times 1$ zero matrix and $\ve{\mathcal{L}^+},\ve{\mathcal{L}^-}\in \Re^{K\times 1}$ with elements $\mathcal{L}^+=\sum_{n=1}^I(x_{n,t}^++l_{Q,t}^+),~\mathcal{L}^-=\sum_{n=1}^I(x_{n,t}^-+l_{Q,t}^-)$, respectively.

We define \eqref{eq:id7} to ensure the continuity of the CES device operation of the following day and to avert its over-charging or over-discharging across $M$ such that,
\begin{equation}
q_0=q_K, \label{eq:id7}
\end{equation}
where $q_K$ is the charge level at the end of day. Readers are referred to \cite{TCNS} for detailed description of the CES model.
\subsection{Energy cost models}
The pricing mechanism of the grid is similar to \cite{MohensianGT} and in particular the unit energy price at time $t$ depends on the total load on the grid at time $t$, $L_t = \sum_{n=1}^Il_{n,t}+l_{{\mathcal N},t}+l_{Q,t}$ where $l_{{\mathcal N},t}$ is the total grid load of the users $\mathcal{N}$ and . Then at time $t$, the unit energy price of the grid is,
\begin{equation}
p_t= \phi_tL_t+ \delta_t, \label{eq:id8}
\end{equation}
where $\phi_t>0$ and $\delta_t>0$. The CES operator also adopts a unit energy price $a_t$ for their energy transactions with the users $\mathcal{P}$ such that any user $n \in \mathcal{P}$ receives $a_tx_{n,t}$ from the CES operator for their selling energy $x_{n,t}$. Then the energy cost of the user $n\in \mathcal{P}$ at time $t$ is,
\begin{equation}
C_{n,t}=p_tl_{n,t}-a_tx_{n,t}. \label{eq:id9}
\end{equation}

The CES operator obtains a revenue from the energy trading with the users $\mathcal{P}$ and the grid that is given by,
\begin{equation}
R=\sum_{t=1}^K\big(-a_t\sum_{n\in \mathcal{P}}x_{n,t}-p_tl_{Q,t}\big). \label{eq:id10}
\end{equation}
Here, we assume that the energy trading between the CES operator and the grid uses the energy rate of the grid.

\section{Energy Trading Stackelberg game }\label{Sec2}
In the energy trading system, the CES operator maximizes their revenue in \eqref{eq:id10} by choosing optimal $a_t$ and $l_{Q,t}$. Following the strategies of the CES operator, each user $n\in \mathcal{P}$ is supposed to minimize their energy cost in \eqref{eq:id9} at each time $t\in M$ by determining optimal $x_{n,t}$. Based on contractual agreements with the system owners, the users $\mathcal{P}$ can individually choose a time $h_n\in \{1,2,\dotsm,K\}$ to start energy trading with the system such that their total daily energy cost is minimized (this process is explained in Section~\ref{Sec3}). After participating at $h_n$, they continue to trade energy for $h_n\leq t \leq K$. Given the opportunity to choose energy trading starting times, the number of active participating users at each time $t$ may not be uniform, and we denote the number of active participating users at time $t$ is $I_t=|\mathcal{P}_t|\leq I$ where $\mathcal{P}_t\subset \mathcal{P}$.

\subsection{Participating Users-Side Analysis}
Using the pricing signal $\ve{a} = [a_1,\dotsm,a_K]$ and the grid energy trading profile $\ve{l_Q}=[l_{Q,1},\dotsm,l_{Q,K}]$ broadcasted by the CES operator, the users $\mathcal{P}_t$ at each time $t\in [1,\dotsm,K]$ minimize their personal energy costs in \eqref{eq:id9}. Let us consider a single time slot $t$ where $I_t~ \geq 2\,$\footnote{\cm{$I_t=1$ implies that there is a single active user who minimizes their energy cost without a game among users $\mathcal{P}$.}}. Then for user $k\in \mathcal{P}_t$, the cost function \eqref{eq:id9} is quadratic with respect to $x_{k,t}$,
\begin{equation}
C_{k,t}=\omega_1x_{k,t}^2+\omega_2x_{k,t}+\omega_3, \label{eq:id11}
\end{equation}
where $\omega_1=\phi_t,~\omega_2=(\phi_t(L_{-k,t}-2s_{k,t})+\delta_t-a_t)$ and $\omega_3=(\phi_ts_{k,t}(s_{k,t}-L_{-k,t})-\delta_ts_{k,t})$ using \eqref{eq id:2} and \eqref{eq:id8}. Here, $L_{-k,t}$ is the total grid energy load at time $t$ excluding the load of the user $k$ and $L_{-k,t}=\sum_{k'\in {\mathcal{P}\backslash k}}l_{k',t}+l_{\mathcal{N},t}+l_{Q,t}$. Clearly, \eqref{eq:id11} is interdependent on each other's behavior and we study the energy trading coordination between the users $\mathcal{P}_t$ using a non-cooperative game $G \equiv \langle \mathcal{P}_t,\mathcal{X},\mathcal{C}\rangle$. Here, $\mathcal{X}=\{ \ve{X}_{1,t},\dotsm,\ve{X}_{k,t},\dotsm,\ve{X}_{{I^t},t}\}$ is the strategy set available to the users $\mathcal{P}_t$ and $\ve{X}_{k,t}$ is the strategy set of the user $k$ subject to \eqref{eq id:4}. $\mathcal{C}$ is the set of cost functions given by $\mathcal{C}=\{C_{1,t},\dotsm,C_{k,t},\dotsm,C_{I^t,t}\}$.

Each user $k\in \mathcal{P}_t$ determines the optimal energy trading amount from $\ve{X}_{k,t}$ such that their energy cost $C_k(x_{k,t},\ve{x}_{-k,t})\equiv C_{k,t}$ is minimized. Here, $\ve{x}_{-k,t}$ denotes the strategy profile of the opponents of the user $k$ that is given by $\ve{x}_{-k,t}=\{x_{1,t},\dotsm\,x_{k-1,t},x_{k+1,t}\dotsm,x_{{I^t},t}\}$. Then the optimization problem of each user $k\in \mathcal{P}_t$ is to find,
\begin{equation}
\tilde{x}_{k,t}=\argmin_{x_{k,t}\in \ve{X_{k,t}}}C_k(x_{k,t},\ve{x}_{-k,t}). \label{eq:id12}
\end{equation}

Note that the game $G$ is similar to the non-cooperative subgame between users in \cite{TCNS}. However, the subsets of players $\mathcal{P}_t$ are not uniform for the game $G$ played at each time $t\in M$ in contrast to \cite{TCNS}. Although the number of players is time-variant, using the same rationale in \cite{TCNS} we can prove that the game $G$ played at any particular time $t$ has a unique Nash equilibrium for any feasible $a_t$ and $l_{Q,t}$. At the Nash equilibrium of the game $G$, the optimal energy trading amount of the user $k$,~$\bar{x}_{k,t}$ can be found by setting the first derivative of \eqref{eq:id11} with respect to $x_{k,t}$ to zero that gives,
\begin{equation}
\frac{\partial C_{k,t}}{\partial x_{k,t}}=2\omega_1 \bar{x}_{k,t}+\omega_2=0. \label{eq:id13}
\end{equation}
Solving \eqref{eq:id13} for all users $\mathcal{P}_t$ simultaneously, we can obtain,
\begin{equation}
\bar{x}_{k,t}=s_{k,t}+\gamma_t, \label{eq:id13b}
\end{equation}
where $\gamma_t=(I_t+1)^{-1}({\phi_t}^{-1}(a_t-\delta_t)-l_{{\mathcal{N}},t}-l_{Q,t}).$

\subsection{CES operator-Side Analysis}
The CES operator also maximizes their revenue in \eqref{eq:id10} by determining optimal $\ve{a}$ and $\ve{l_Q}$. By substituting \eqref{eq:id13b} in \eqref{eq:id10}, we can write the objective of the CES operator as,
\begin{equation}
[\tilde{\ve{a}},\tilde{\ve{l_Q}}]=\argmax_{\ve{a},\ve{l_Q}\in\mathcal{Q}}{\sum_{t=1}^K (\lambda_1 a_t^2+\lambda_2 a_t+\lambda_3 l_{Q,t}^2+\lambda_4l_{Q,t})}, \label{eq:id14}
\end{equation}
where $\lambda_1=-I_t(I_t+1)^{-1}{\phi_t}^{-1}$, $\lambda_2=I_t(I_t+1)^{-1}(l_{\mathcal{N},t}+{\phi_t}^{-1}\delta_t)-\sum_{k=1}^{I_t}{s_{k,t}}$, $\lambda_3=-\phi_t(I_t+1)^{-1}$, and $\lambda_4 =-(I_t+1)^{-1}(\phi_t l_{\mathcal{N},t}+\delta_t)$. $\mathcal{Q}$ is the strategy set available to the operator subject to \eqref{eq:id6} and \eqref{eq:id7}. There is a unique solution for the objective function of the CES operator, since \eqref{eq:id14} is strictly concave because of the negative definite Hessian matrix with respect to all feasible $\ve{a},~\ve{l_Q}$ and the strategy set $\mathcal{Q}$ is convex due to linear constraints \eqref{eq:id6} and \eqref{eq:id7}.

\subsection{Stackelberg Equilibrium}
The CES operator first sets optimal $[\ve{a}, \ve{l_Q}]$ to maximize \eqref{eq:id10} and broadcasts them to the users $\mathcal{P}\equiv \{\mathcal{P}_1\cup\dotsm\cup\mathcal{P}_K\}$. Then the users $\mathcal{P}_t$ at each time $t\in M$ follow these signals to find optimal $x_{k,t}$ by playing the game $G$. We model this hierarchical interaction between the CES operator and the users $\mathcal{P}$ using a non-cooperative Stackelberg game $\Xi$. In the game $\Xi$, \textit{players} are the CES operator and the users $\mathcal{P}$ where the CES operator is the leader and the users $\mathcal{P}$ are the followers. As the \textit{strategies}, the CES operator determines $[\ve{a}, \ve{l_Q}]\in \mathcal{Q}$ to maximize \eqref{eq:id10} and at time $t$, user $k\in \mathcal{P}_t$ selects $x_{k,t}\in \ve{X}_{k,t}$ to minimize cost in \eqref{eq:id9}. The \textit{utilities} are as defined in \eqref{eq:id10} for the CES operator and \eqref{eq:id9} for the user $k\in \mathcal{P}_t$.
\begin{definition} Let $\ve{\hat\rho}\equiv[\ve{\hat a},\ve{\hat l_Q}]$ be the solution of \eqref{eq:id14} and $\ve{\hat X}\equiv \{ [\ve{{\hat x_1}}]^T,\dotsm,[\ve{\hat x_K}]^T\}$ where $\ve{\hat x_t}$ be the solution of the game $G$ at time $t\in M$. Then the point $[\ve{\hat \rho},\ve{\hat X}]$ is a Stackelberg equilibrium if and only if,
\begin{equation}
 \begin{split}
R(\ve{\hat {X}},\ve{\hat \rho}) \geq R(\ve{\hat X},\ve{\rho}),~\forall \ve{\rho} \in \mathcal{Q},
\end{split}\label{eq:id16}
\end{equation}
\begin{gather}
C_k(\ve{\hat x_t},~\ve{\hat\rho}) \leq C_k(x_{k,t},\ve{\hat x_{-k,t}}, \ve{\hat \rho}), \nonumber \\
\forall k\in \mathcal{P}_t,~\forall x_{k,t}\in \ve{X_{k,t}},~\forall t\in M.\label{eq:id15}
\end{gather}
\end{definition}

\begin{proposition}
The game $\Xi$ has a unique Stackelberg equilibrium.
\end{proposition}

\begin{proof}
The game $G$ played at any time $t$ has a unique Nash equilibrium for feasible $a_t$ and $l_{Q,t}$. Further, the revenue maximization of the CES operator in \eqref{eq:id14} has a unique solution. Hence, the game $\Xi$ converges to a unique Stackelberg equilibrium once the CES operator obtains optimal strategy $\ve{\hat \rho}$ while the users $\mathcal{P}$ attain their K-tuple of unique Nash equilibrium solutions $\ve{\hat X}$.
\end{proof}

\section{Participation-Time Selection Game}\label{Sec3}
Before the Stackelberg game in Section~\ref{Sec2} takes place, the users $\mathcal{P}$ individually select optimal times $h_n$ to start energy trading such that their total daily energy costs are minimized. Similar to \cite{Wang}, to explicitly study such user behavior with respect to choosing to participate across time in our system, we develop a non-cooperative game $\Gamma$ between the users $\mathcal{P}$ that has the strategic form $\Gamma =\langle \mathcal{P},\ve{H},\ve{U}\rangle$ and study it under expected utility theory and prospect theory. Here, $\ve{H}$ is the set of available strategies to the users $\mathcal{P}$ i.e., energy trading starting times. $\ve{H}=\{\ve{H}_n\}_{n\in \mathcal{P}}$ where $\ve{H}_n=\{1,\dotsm,K\}$; $\forall n\in \mathcal{P}$. $\ve{U}=\{U_n\}_{n\in \mathcal{P}}$ is the set of cost functions that captures the daily energy costs for each user $n\in \mathcal{P}$. Note that the Stackelberg equilibrium described in Section~\ref{Sec2} depends on temporal distribution of the users $\mathcal{P}_t$ which is a result of how the users $\mathcal{P}$ begin energy trading across time. At the Stackelberg equilibrium corresponding to an action (i.e., energy trading starting time) profile $\ve{h}=\{h_n,\ve{h}_{-n}\}=\{h_1,\dotsm,h_I\}$ of the users $\mathcal{P}$ where $h_n\in \ve{H}_n$, the daily energy cost of user $n$, $U_n$ is,
 \begin{equation}
U_n(\ve{h})=\sum_{t=1}^K\hat p_t(\ve{h})\hat l_{n,t}(\ve{h})-\hat a_t(\ve{h})\hat x_{n,t}(\ve{h}). \label{eq:id17}
\end{equation}
Here, $\hat p_t(\ve{h}),~\hat a_t(\ve{h}),~\hat l_{n,t}(\ve{h}),~\hat x_{n,t}(\ve{h})$ are the grid price, CES energy price, user $n$'s grid load and their CES energy trading amount at the Stackelberg equilibrium obtained for $\ve{h}$, respectively. Note that $\ve{h}_{-n}$ is the action profile of the users $\mathcal{P}$ except user $n$. In the game $\Gamma$, each user $n\in \mathcal{P}$ chooses an energy trading starting time $h_n$ for given $\ve{h}_{-n}$ such that their energy cost in \eqref{eq:id17} is minimized.

\begin{remark}
 After the users $\mathcal{P}$ decide to participate as per optimal energy trading starting times determined by playing the game $\Gamma$, the Stackelberg energy trading in Section~\ref{Sec2} takes place that ultimately achieves a Stackelberg equilibrium.
\end{remark}

 In the long run, the users $\mathcal{P}$ may change their behavior with respect to choosing an energy trading starting time. Hence, we investigate a solution for the non-cooperative game $\Gamma$ that captures empirical frequencies of actions followed by the users $\mathcal{P}$. The straightforward interpretation is that each user $n\in \mathcal{P}$ assigns a probability for each action in $\ve{H}_n$. In such a paradigm, users face uncertainty to make decisions and we characterize solutions for the game $\Gamma$ based on mixed strategies under two different user-behavioral models: expected utility theory and prospect theory.

 \subsection{Energy Trading Under Expected Utility Theory}
Under the notion of mixed strategies, each user $n\in \mathcal{P}$ determines the optimal probability distribution over the actions in $\ve{H}_n$ to minimize expected daily energy cost. Here, we explore how the users $\mathcal{P}$ decide probabilities of energy trading starting times according to expected utility theory assuming that all users make rational choices by objectively viewing their opponents' behavior. According to the theory, the expected daily energy cost of the user $n$ can be given as,
 \begin{equation}
E_n^{EUT}(\ve{y})=\sum_{\ve{h}\in \ve{H}}U_n(\ve{h})\prod_{r=1}^Iy_r(h_r), \label{eq:id18}
\end{equation}
where $\ve{y}=\{\ve{y}_n,\ve{y}_{-n}\}$,~$\ve{y}_n=[y_n(1),\dotsm,y_n(K)]$ and $y_n(h_n)$ is the probability of choosing $h_n$ by the user $n$. $\ve{y}_{-n}$ is the probabilities of the users $\mathcal{P}$ except user $n$.

The intuition behind the cost in \eqref{eq:id18} relies on the assumption that the user $n$ assesses their neighbours' empirical frequencies of actions identical to their objective probabilities of choosing actions. However, this generalization may not be valid in the real world as people overweight outcomes with low probabilities and underweight outcomes with high probabilities. These observations are clearly explained under prospect theory \cite{Kahneman}.

\subsection{Energy Trading under Prospect Theory}
In practice, the users $\mathcal{P}$ may subjectively evaluate their neighbors' actions to minimize energy costs. This characteristic is more realistic than assuming users act rationally and perceive their neighbors' behavior objectively \cite{Kahneman}. In this regard, we study actual user behavior as to when they select their energy trading starting time using prospect theory.

To this end, probability weighting functions are used to model the subjective behavior of users when they make decisions under risk and uncertainty. In this regard, the probability weighting function $w_n(y)$ implies the subjective evaluation of the user $n$ about an outcome with $y$ probability. We use the Prelec function \cite{prelec} to model the subjective perceptions of users on each other's behavior that is given by,
 \begin{equation}
w_n(y)=\exp(-(-\ln y)^{\alpha_n});~0<\alpha_n\leq1. \label{eq:id19}
\end{equation}
Here, $\alpha_n$ is a parameter that decreases as the user's subjective evaluation deviates from the objective probability. If the user's subjective and objective probabilities are equal, then $\alpha_n=1$ and this corresponds to expected utility theory. Assuming that the subjective probabilities of user $n\in \mathcal{P}$ about their own actions are equal to their objective probabilities, the expected daily energy cost of user $n$ under prospect theory is,
 \begin{equation}
 E_n^{PT}(\ve{y})=\sum_{\ve{h}\in \ve{H}}U_n(\ve{h})y_n(h_n)\big(\prod_{r\in \mathcal{P}\backslash n}^{I-1}w_n(y_r(h_r))\big).\label{eq:id20}
\end{equation}

\subsection{$\epsilon$-Nash Equilibria}
After defining the expected daily costs of the users $\mathcal{P}$, we now analyze the solutions for the game $\Gamma$ played under expected utility theory and prospect theory. Due to computational usefulness \cite{gametheoryessentials}, here we study the existence of $\epsilon-$Nash equilibria. For the game $\Gamma$, a mixed strategy profile $\ve{y^*}\equiv \{\ve{y}_n^*,\ve{y}_{-n}^*\}$ is an $\epsilon-$Nash equilibrium if it satisfies, \begin{equation}
E_n(\ve{y}_n^*,\ve{y}_{-n}^*)\leq E_n(\ve{y}_n,\ve{y}_{-n}^*)+\epsilon;~\forall \ve{y}_n\in \mathcal{Y}_n, \forall n\in \mathcal{P},\label{eq:id21}
\end{equation}
where $\mathcal{Y}_n$ is the set of all mixed strategy profiles over $\ve{H_n}$ and $\epsilon>0$. In general, $\epsilon-$Nash equilibria always exist \cite{gametheoryessentials} and for the game $\Gamma$, we are interested to find $\epsilon$-Nash equilibrium located close to a mixed strategy Nash equilibrium under both expected utility theory and prospect theory. We use the iterative algorithm proposed in \cite{Wang} that was proved to converge to an $\epsilon$-Nash equilibrium close to a mixed strategy Nash equilibrium under both expected utility theory and prospect theory. In summary, the algorithm is given by,
\begin{equation}
\ve{y}_n^{(i+1)}=\ve{y}_n^{(i)}+\frac{\eta}{i}(\ve{v}_n^{(i)}-\ve{y}_n^{(i)}),\label{eq:id22}
\end{equation}
where $i$ is the iteration number,~$0<\eta<1$ is the inertia weight. $\ve{v}_n^{(i)}=\{v_n^{(i)}(h_{n,1}),\dotsm,v_n^{(i)}(h_{n,K})\}$ of which,
\begin{equation}
v_n^i(h_{n,t})=\begin{cases}
                 1,~\text{if}~h_{n,t}=\operatornamewithlimits{argmin}\limits_{h_n\in \ve{H}_n}e_n(h_n,\ve{y}_{-n}^{(i-1)}), \\
                 0,~\text{otherwise},
                   \end{cases}\label{eq:id23}
\end{equation}
where $e_n(h_n,\ve{y}_{-n}^{(i-1)})$ is the expected cost when the user $n$ selects the pure strategy $h_n$ in response to the mixed strategies of other players at iteration $(i-1)$ i.e., $\ve{y}_{-n}^{(i-1)}$. Note that for prospect theory, $\ve{y}_{-n}^{(i-1)}$ considers the weighted probabilities of other users' mixed strategies at $(i-1)$.
\begin{remark}
As the algorithm converges, $\epsilon$-Nash equilibrium with respect to strategy profile $\ve{y}$ is obtained under both expected utility theory and prospect theory.
\end{remark}

Given the equilibrium probabilities of participating-time decisions of the users $\mathcal{P}$, we can define the expected revenue of the CES operator under both prospect theory and expected utility theory. In this regard, if $\ve{y}_{EUT}^*$ and $\ve{y}_{PT}^*$ are the $\epsilon$-Nash equilibriums under expected utility theory and prospect theory, respectively, then the subsequent expected daily CES revenue $W$ in each case can be obtained by,
\begin{equation}
W=\sum_{\ve{h}\in\ve{H}}R(\ve{h})\prod_{r=1}^Iy_r^*(h_r), \label{eq:id24}
\end{equation}
where $R(\ve{h})$ is the CES revenue as per \eqref{eq:id10} at the Stackelberg equilibrium corresponds to $\ve{h}$, $y_r^*(h_r)\in \ve{y}_{EUT}^*$ for expected utility theory and $y_r^*(h_r)\in \ve{y}_{PT}^*$ for prospect theory.

\section{Simulation Results}\label{Sec4}
In simulations, we consider real data of average PV power and user demand of the Western Power Network in Australia on a summer day \cite{WPNdata} (see Fig.~\ref{fig:PVdemand}) and we assume that all users have power profiles same to these average profiles. Further, $K=24$, $\Delta=1 \text{h}$, $B=80$ kWh,  $q_0=20$ kWh, $\tau=0.9^{(1/48)}$, $\beta^+=0.9$ and $\beta^-=1.1$ \cite{Atzeni}. Peak hours of the grid are between 16.00 and 23.00 and we select $\phi_t$ such that $\phi_{\text{peak}}=1.5~\phi_{\text{off-peak}}$. We choose $\phi_{\text{peak}}$ such that the predicted grid price range is same to the reference time-of-use price range in \cite{ausgrid} and $\delta_t$ is set to a constant such that the average predicted grid price is equal to the average reference price. The community has 10 households where 6 users are participating users $\mathcal{P}$ in the system. The allowable energy trading starting times for the users $\mathcal{P}$ are 01.00, 12.00 and 17.00 so that $\ve{H}_n=\{1,12,17\}$. For comparisons, we use a baseline without a CES device where the users $\mathcal{P}$ trade energy directly with the grid that uses the same energy cost model. For the algorithm, we use $\ve{y}_n^{(0)}=[0.3,0.3,0.4];~\forall n\in \mathcal{P}$ and $\eta = 0.7$.

\begin{figure}[t!]
\centering
\includegraphics[width=0.79\columnwidth]{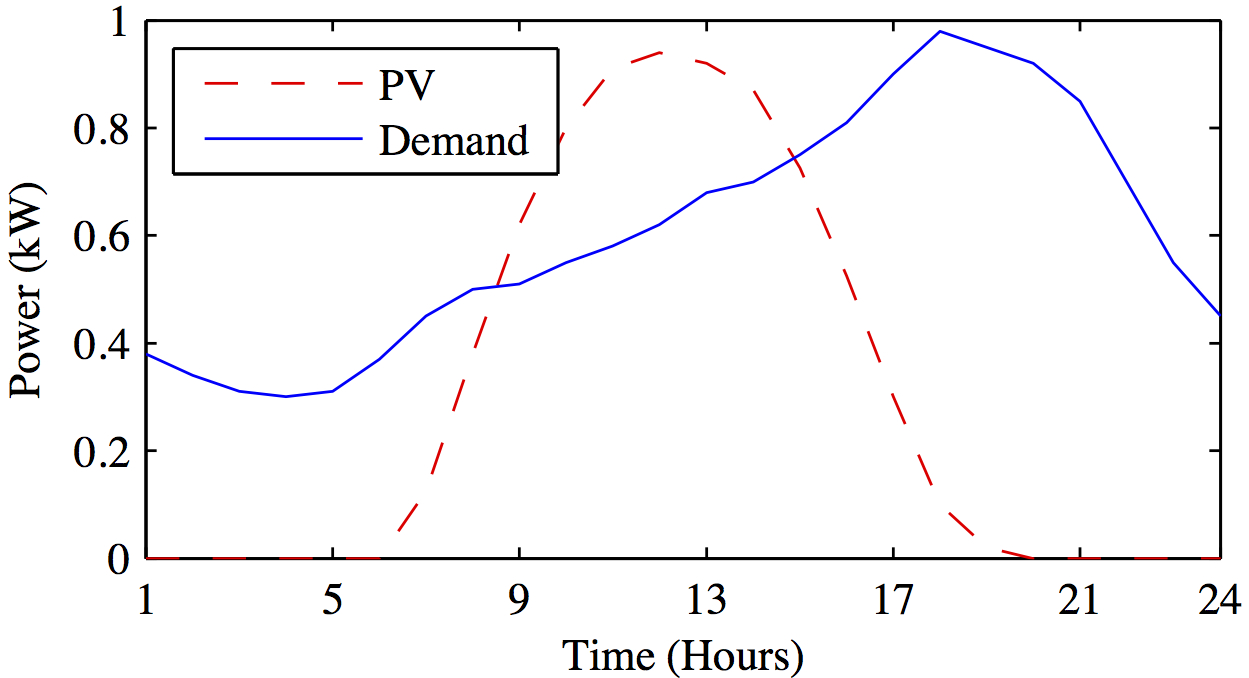}
\caption{Average PV power and user electricity demand.}
\label{fig:PVdemand}
\end{figure}

Fig.~\ref{fig:ExpCostSav} illustrates expected cost savings of the users $\mathcal{P}$ under expected utility theory, and under prospect theory for three different $\alpha \in (0,1]~(\text{i.e.,}~0.7,~0.4~\text{and}~0.1)$\footnote{From \eqref{eq:id19}, as $\alpha$ tends to $0$ users become more subjective deviating from the objective evaluation assumption in expected utility theory.} assuming $\alpha_n=\alpha;~\forall n\in \mathcal{P}$. Here, cost savings are calculated compared to the baseline. When $\alpha = 0.7$, and even when $\alpha=0.4$ with significant non-ideal behavior, the expected cost savings remained almost 28\% under both models because for all users, participation probabilities at each time in $\ve{H}_n$ using prospect theory do not significantly deviate from those obtained under expected utility theory as shown in Table~\ref{table 1}. When $\alpha = 0.1$, the participation probabilities at $h_n=1$ are significantly increased for the fourth and fifth users compared to those predicted using expected utility theory (see Table~\ref{table 1}). As a result, the expected cost savings reduced from 28\% to 21.5\% for all users.

\begin{figure}[t!]
\centering
\includegraphics[width=0.79\columnwidth]{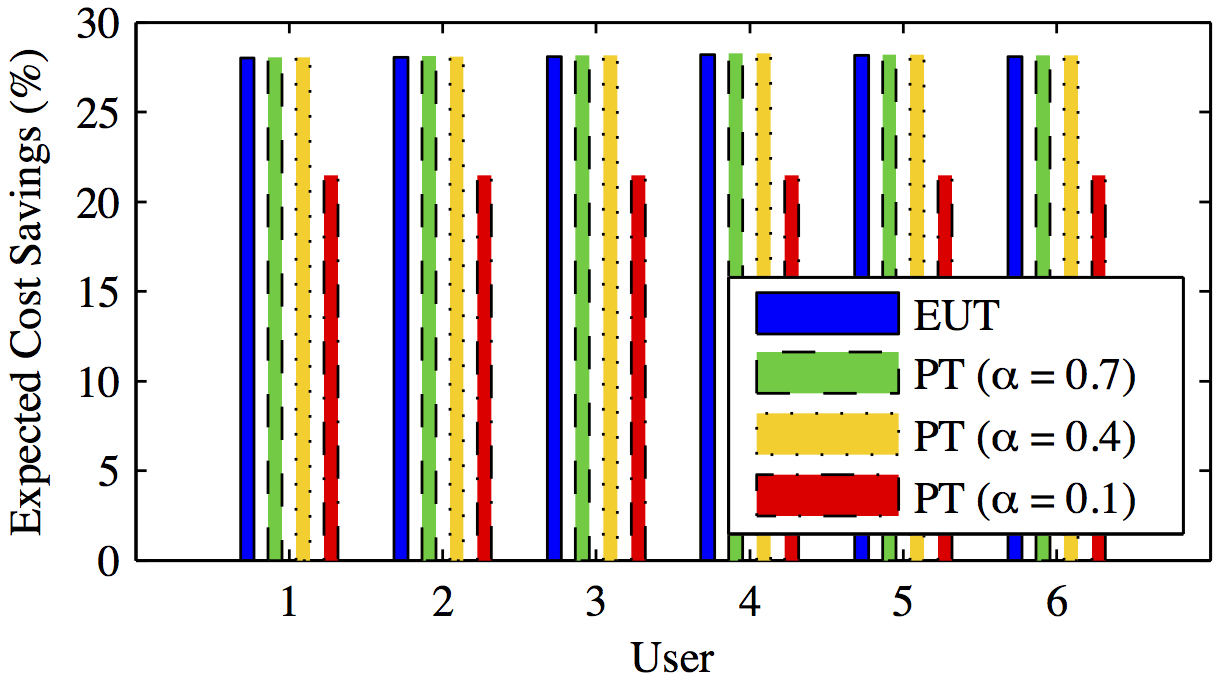}
\caption{Expected user cost savings under expected utility theory (EUT) and prospect theory (PT).}
\label{fig:ExpCostSav}
\end{figure}

\begin{figure}[b!]
\centering
\includegraphics[width=0.79\columnwidth]{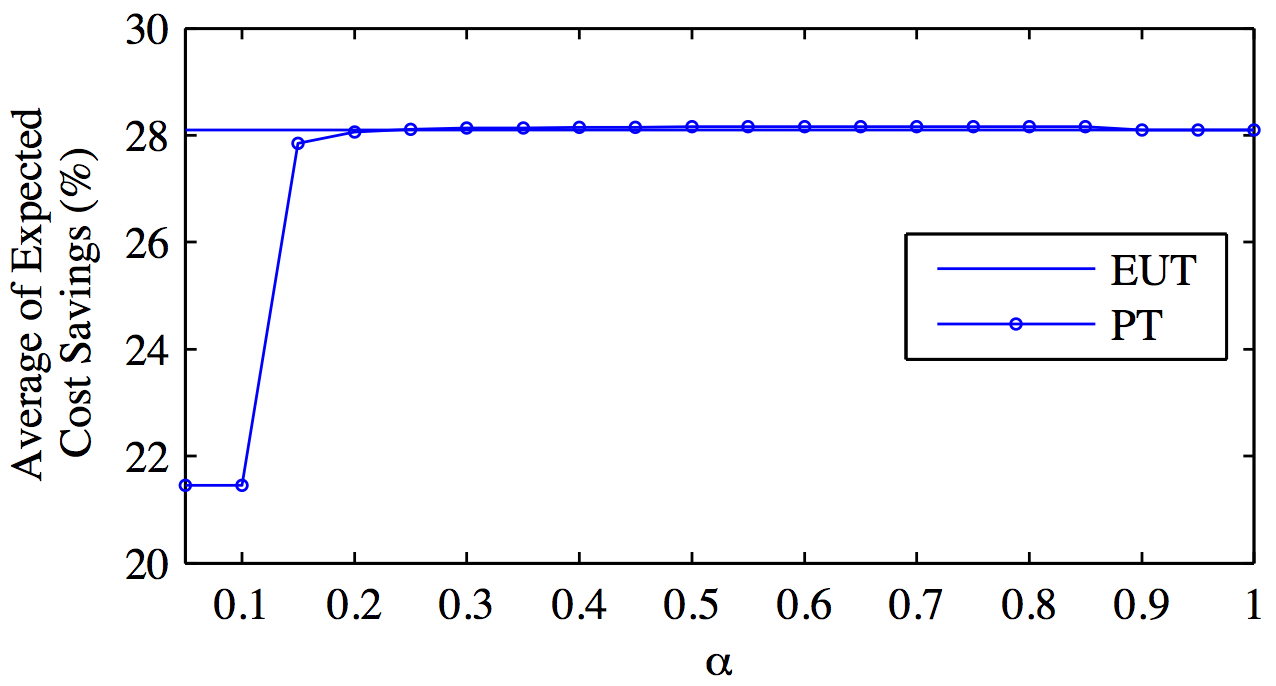}
\caption{Average of expected user cost savings with different $\alpha$.}
\label{fig:AvgExpCostSav}
\end{figure}

\begin{figure}[bt!]
\centering
\includegraphics[width=0.79\columnwidth]{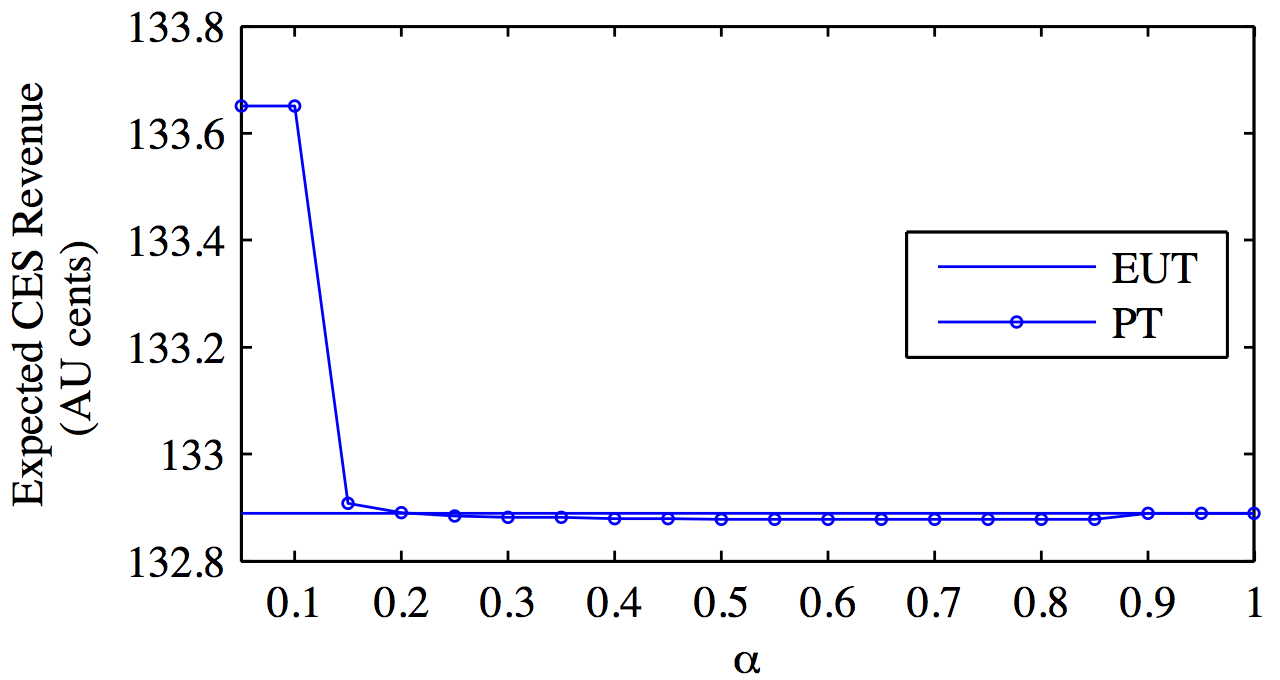}
\caption{Expected CES operator revenue with different $\alpha$. }
\label{fig:CES Rev}
\end{figure}

\begin{figure}[bt!]
\centering
\includegraphics[width=0.79\columnwidth]{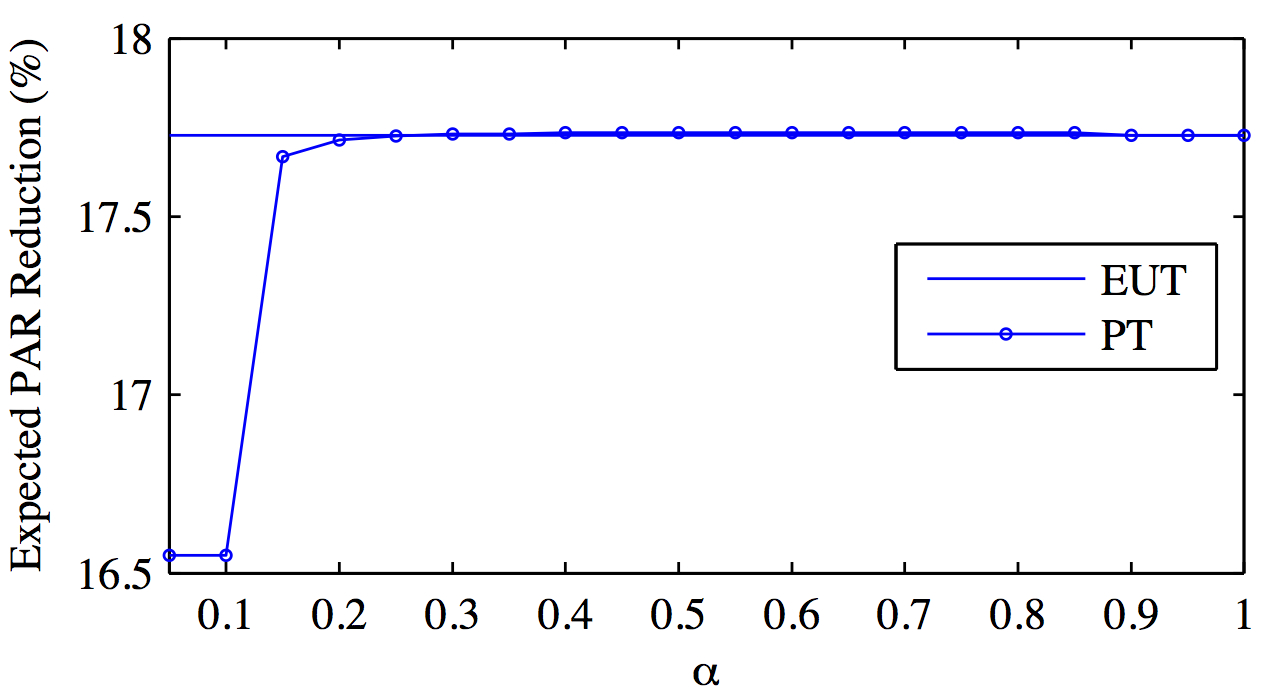}
\caption{Expected peak-to-average ratio (PAR) reduction with different $\alpha$. }
\label{fig:PARred}
\end{figure}

\begin{table*}[bt!]
\renewcommand{\arraystretch}{1.3}
\caption{Participation Probabilities of users $\mathcal{P}$ for $h_n\in \ve{H}_n\equiv \{1,12,17\}$ when $\alpha = 0.7,~0.4,~0.1$}
\label{table 1}
\centering
\begin{tabular}{|c|c|c|cIc|c|cIc|c|cIc|c|c|}
\hline
  {}&\multicolumn{3}{cI}{EUT}&\multicolumn{3}{cI}{PT $(\alpha=0.7)$}&\multicolumn{3}{cI}{PT $(\alpha=0.4)$}&\multicolumn{3}{c|}{PT $(\alpha=0.1)$} \\
  \cline{2-13}
 User&$h_n$~=~1&$h_n$~=~12&$h_n$~=~17&$h_n$~=~1&$h_n$~=~12&$h_n$~=~17&$h_n$~=~1&$h_n$~=~12&$h_n$~=~17&$h_n$~=~1&$h_n$~=~12&$h_n$~=~17\\
\hline
1 & 0.9966&0.0005&0.0029 & 0.9988&0.0005&0.0007 & 0.9989&0.0005&0.0006 & 0.9979&0.0009&0.0012\\
\hline
2 & 0.9966&0.0005&0.0029 & 0.9988&0.0005&0.0007 & 0.9989&0.0005&0.0006 & 0.9979&0.0009&0.0012\\
\hline
3 & 0.9966&0.0005&0.0029 & 0.9988&0.0005&0.0007 & 0.9989&0.0005&0.0006 & 0.9979&0.0009&0.0012\\
\hline
4 & 0.0070&0.9924&0.0006 &  0.0076&0.9918&0.0006 & 0.0070&0.9924&0.0006 & 0.9979&0.0009&0.0012\\
\hline
5 & 0.0070&0.0005&0.9925 & 0.0076&0.0005&0.9919 & 0.0095&0.0005&0.9900 & 0.9979&0.0009&0.0012\\
\hline
6 & 0.9966&0.0005&0.0029 & 0.9988&0.0005&0.0007 & 0.9989&0.0005&0.0006 & 0.9979&0.0009&0.0012\\
\hline
\end{tabular}
\end{table*}

Fig.~\ref{fig:AvgExpCostSav}, Fig.~\ref{fig:CES Rev} and Fig.~\ref{fig:PARred} depict the variations in different aspects of system performance across the range of possible $\alpha$ values. Here, larger $\alpha$ tending to 1 reflects that the users behave closer to the rationality assumption in expected utility theory, and smaller $\alpha$ tending to 0 implies that their evaluations of opponents' actions are more distorted from that of expected utility theory. Fig.~\ref{fig:AvgExpCostSav} shows that under expect utility theory, the average of expected cost savings of the users achieved by participating in the system is 28.1\%. On the other hand, even if the users' weighting effects on their opponents' actions are getting larger, i.e., when $\alpha$ is getting smaller, the expected cost savings will not significantly fluctuate and remain almost at 28\% except for $0<\alpha\leq0.1$. Fig.~\ref{fig:CES Rev} shows that, when $\alpha>0.15$, expected revenue for the CES operator retains nearly unchanged compared to the expected revenue calculated under expected utility theory. In terms of demand-side management of the grid, the expected peak-to-average ratio reduction compared to the baseline will not change notably from the peak-to-average ratio reduction predicted using expected utility theory when $\alpha>0.15$. This is because as shown in Table~\ref{table 1}, for $\alpha>0.15$, users' prospect theoretic probabilities of participation at each time remain almost the same as those in expected utility theory. When $0<\alpha\leq0.1$, the fourth and fifth users will more likely to start energy trading from the beginning under prospect theory, which is not the case under expected utility theory. However, this behavioral change will only reduce the expected peak-to-average ratio reduction from 17.7\% to 16.55\% (see Fig.~\ref{fig:PARred}).

\section{Conclusion}\label{Sec5}
In this paper, we have studied effects of realistic, non-ideal, behavior of users, with respect to choosing energy trading starting times, on a game-theoretic demand-side management energy trading system between a community energy storage (CES) device and users. First, we have developed the non-cooperative Stackelberg game to study the energy trading interaction between the users and the CES operator based on users' decisions as to whether to participate across time. Next we have studied a non-cooperative game to explore how the users make decisions to participate in the above energy trading system under two user-behavioral models: prospect theory and expected utility theory. Simulation results show that the benefits of the energy trading system are robust to users' strategies of participating-time that significantly deviate from complete rationality. We postulate that the energy trading system can be scaled to any number of participating users and present similar performance trends.


 \newcommand{\noop}[1]{}

\end{document}